\documentclass{article}

\usepackage[preprint, nonatbib]{neurips_2025}

\usepackage[utf8]{inputenc} 
\usepackage[T1]{fontenc}    
\usepackage{hyperref}       
\usepackage{url}            
\usepackage{booktabs}       
\usepackage{amsfonts}       
\usepackage{nicefrac}       
\usepackage{microtype}      
\usepackage{xcolor}         
\usepackage{amsmath}
\usepackage{amsthm}
\usepackage{amssymb}
\usepackage{graphicx}
\usepackage{algorithm}
\usepackage{algpseudocode}
\usepackage{tablefootnote}

\usepackage{biblatex}
\bibliography{references}

\newtheorem*{assumption*}{Assumption}
\newtheorem*{proposition*}{Proposition}
\newtheorem*{definition*}{Definition}
\newtheorem*{theorem*}{Theorem}

\title{Making Multi-Axis Models Robust to Multiplicative Noise: How, and Why?}

\author{%
  Bailey Andrew \\
  School of Computer Science\\
  University of Leeds\\
  Woodhouse, Leeds, UK; LS2 9JT \\
  \texttt{sceba@leeds.ac.uk} \\
  \And
  David R. Westhead \\
  School of Molecular and Cellular Biology \\
  University of Leeds \\
  \And
  Luisa Cutillo \\
  School of Mathematics \\
  University of Leeds \\
}

\newcommand{\rows}{\mathrm{rows}}
\newcommand{\cols}{\mathrm{cols}}
\newcommand{\cells}{\mathrm{cells}}
\newcommand{\genes}{\mathrm{genes}}

\begin{document}

\maketitle

\begin{abstract}
  In this paper we develop a graph-learning algorithm, MED-MAGMA, to fit multi-axis (Kronecker-sum-structured) models corrupted by multiplicative noise.  This type of noise is natural in many application domains, such as that of single-cell RNA sequencing, in which it naturally captures technical biases of RNA sequencing platforms.  Our work is evaluated against prior work on each and every public dataset in the Single Cell Expression Atlas under a certain size, demonstrating that our methodology learns networks with better local and global structure.  MED-MAGMA is made available as a Python package (\texttt{MED-MAGMA}).
\end{abstract}

\section{Introduction}
\label{sec:med-magma-introduction}

This paper deals with the problem of how to make Kronecker-structured (or `multi-axis') methods robust to a form of noise common in applied domains such as those using single-cell RNA sequencing data (scRNA-seq; matrix-variate data in which rows represent cells, columns represent genes, and entries represent the amount a given gene is expressed in a given cell).  Multi-axis models aim to simultaneously learn matrices $\mathbf{\Psi}_\rows, \mathbf{\Psi}_\cols$ that encode dependency networks (`graphs') between the rows and the columns, respectively.  For scRNA-seq, this would yield both cell and gene networks, useful for tasks such as cell subtyping and gene regulatory network inference.

In particular, we aim to fit models of the following form for matrix-variate data, where $\zeta$ is some function whose role is to mix the $O(d_\rows)$ row and $O(d_\cols)$ column dependencies into a network of $O(d_\rows d_\cols)$ matrix entry dependencies.

\begin{align*}
    \mathrm{latent\text{ }data} &\sim \mathcal{N}\left(\mathbf{0}, \zeta\left(\mathbf{\Psi}_\rows, \mathbf{\Psi}_\cols\right)^{-1} \right) \\
    \mathrm{observed\text{ }data}_{ij} &= r_i^\rows r_j^\cols \left(\mathrm{latent\text{ }data}_{ij}\right) \tag{multiplicative noise}\\
    r_i^\rows, r_j^\cols &\sim \text{any strictly positive distribution}
\end{align*}

There are many methods to fit the above model without noise under a variety of choices of $\zeta$, which we will discuss in Section \ref{sec:med-magma-background}.  However, no such method exists to fit the distributions under this noise regime.  This is unfortunate, as not only is this type of noise common, but taking it into account also dramatically increases the flexibility of the model.

In scRNA-seq, this type of noise arises from two simple observations.  Some genes are more likely to be measured than others due to biases such as GC-bias \parencite{dohm_substantial_2008, benjamini_summarizing_2012}, and cells have less measured gene expression than reality due to the incompleteness of the measuring process.  If the true measurement of a gene $g$ in cell $c$ is $x$, but only $0 \leq r_c^\cells \leq 1$ of the expression in that cell is collected, and gene $g$ is $r_g^\genes $ times as likely to be collected than the other genes, then we should expect to observe a measurement of $r_g^\genes r_c^\cells x$ rather than $x$ itself.

Even in cases where this noise type does not exist, the flexibility of models accounting for this noise is beneficial.  This is because our model is a natural generalization of the class of elliptical distributions to the matrix-variate case.  If a dataset follows the model below, then it has an elliptical distribution:

\begin{align*}
    \mathrm{latent\text{ }data} &\sim \mathcal{N}\left(\mathbf{0}, \mathbf{\Omega}^{-1} \right) \\
    \mathrm{observed\text{ }data} &= r \cdot \left(\mathrm{latent\text{ }data}\right) \\
    r &\sim \text{any strictly positive distribution}
\end{align*}

Elliptical distributions encompass many other distributions, such as the normal distribution (when $r=1$), the multivariate t distribution (when $r\sim\sqrt{\frac{\nu}{\chi^2_\nu}}$), and the multivariate Laplace distribution (when $r^2 \sim \mathrm{Exponential}(1)$).  In fact, any distribution whose density has \textit{elliptical contours} can, under mild conditions, be expressed in this form.

Most multi-axis models either assume the dataset is Gaussian, or make the weaker `Gaussian copula' assumption (i.e. arbitrary marginals but the variables `interact Gaussian-ly'); the latter is often fit via the use of the nonparanormal skeptic \parencite{liu_nonparanormal_2012}.  Gaussian copulas are restrictive in terms of the types of dependencies they allow.  They do not allow `tail dependence' - extreme events are always asymptotically uncorrelated.  However, Figure \ref{fig:med-magma-tail-dependence} implies that tail dependence \textit{does} exist in scRNA-seq data.  Elliptical distributions allow for such extremal dependencies.

\begin{figure}[h!]
    \centering
    \includegraphics[width=\linewidth]{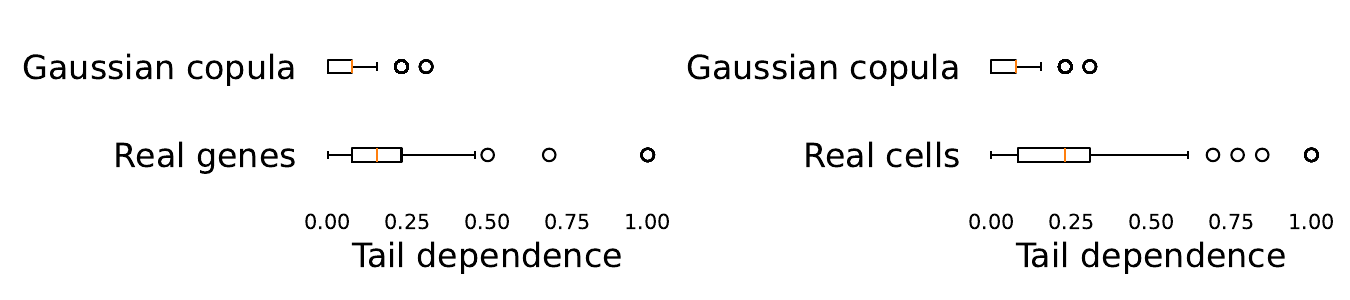}
    \vspace{-18pt}
    \caption{Estimated tail dependence across 1000 randomly selected pairs of genes (left) and cells (right) from the E-MTAB-7249 scRNA-seq dataset discussed in Section \ref{sec:med-magma-specific}, compared to random data generated from a Gaussian copula.}
    \label{fig:med-magma-tail-dependence}
\end{figure}

Thus, it should be clear why our proposed model is useful: it is robust to a natural type of noise, and strongly increases the expressive power of multi-axis models.  In this paper, we will produce an algorithm (MED-MAGMA: the \textbf{m}atrix \textbf{e}lliptical \textbf{d}istribution's \textbf{m}ulti-\textbf{a}xis \textbf{g}raphical \textbf{m}odelling \textbf{a}lgorithm) to learn $\mathbf{\Psi}_\rows, \mathbf{\Psi}_\cols$ under the presented model (Chapter \ref{sec:med-magma-methodology}) and in Chapter \ref{sec:med-magma-results} comprehensively validate its performance across every sub-1000-cell dataset in the public Single Cell Expression Atlas database \parencite{madrigal_expression_2026}.  In the next section we give the motivation and history behind such models.

\section{Background}
\label{sec:med-magma-background}

Multi-axis models were introduced to solve the problem of inference when samples are not independent and identically distributed (i.i.d.).  If one is willing to assume that the dataset $\{\mathbf{d}_i\}$ \textit{is} i.i.d, then a simple formulation would be:

\begin{align*}
    \mathbf{d}_i \sim \mathcal{N}\left(\mathbf{0}, \mathbf{\Psi}_\cols^{-1}\right)
\end{align*}

In this case, $\mathbf{\Psi}_\cols$ encodes the graph of conditional dependencies:

\begin{align*}
    \mathbf{\Psi}^\cols_{ij} = 0 \iff \text{feature (column) } i \text{ is conditionally independent from feature } j
\end{align*}

Two variables being conditionally independent means that, after conditioning out the rest of the dataset, they are statistically independent from one another.  In other words, conditional dependencies intuitively capture the concept of `direct correlations': if $A$ is correlated with $B$, which in turn is correlated with $C$, then $A$ and $C$ will almost always also be correlated to some degree, via their indirect connection through $B$.  However, $A$ and $C$ will not be \textit{conditionally dependent}, as their indirect connection will have been conditioned out.  This connection motivates the use of conditional dependencies: knowledge of direct effects is much more valuable than indirect effects.

For more complex datasets, the i.i.d. assumption breaks down.  In particular, the `samples' of a scRNA-seq dataset are cells of various different cell types, which had been interacting with other cells in the dataset for their entire life-cycle, until being destroyed and collected for data analysis.  In this scenario, not only are the cells not independent, but the dependencies between the cells are also of interest in their own right.  Compiling the dataset into a $d_\rows \times d_\cols$ matrix $\mathbf{D}$, we could address this by using the following distribution, where $\mathrm{vec}$ is the vectorization operation that stacks the columns of a matrix into a single vector:

\begin{align*}
    \mathrm{vec}\left[\mathbf{D}\right] \sim \mathcal{N}\left(\mathbf{0}, \mathbf{\Omega}^{-1}\right)
\end{align*}

However, $\mathbf{\Omega}$ would be a $d_\rows d_\cols \times d_\rows d_\cols$ matrix capturing the dependencies between (cell, gene) pairs - it is computationally intractible to store and statistically intractible to learn.  What we want is to decompose $\mathbf{\Omega}$ into tractable chunks directly representing the rows and columns: $\mathbf{\Omega} = \zeta\left(\mathbf{\Psi}_\rows, \mathbf{\Psi}_\cols\right)$.  Rather than $O(d_\rows^2 d_\cols^2)$ parameters, this framework only requires $O(d_\rows^2 + d_\cols^2)$.  There are many choices for $\zeta$, but they are all based on the Kronecker product $\otimes$:

\[
\left[\begin{smallmatrix} a_{11} & a_{12} & \cdots \\ a_{21} & \cdots & \\ \cdots & & \end{smallmatrix}\right] \otimes \mathbf{B} = \left[\begin{smallmatrix} a_{11}\mathbf{B} & a_{12}\mathbf{B} & \cdots \\ a_{21}\mathbf{B} & \cdots & \\ \cdots & & \end{smallmatrix}\right]
\]


The following are the choices already considered in the literature:

\begin{align*}
    \zeta\left(\mathbf{\Psi}_\rows, \mathbf{\Psi}_\cols\right) &= \mathbf{\Psi}_\cols \otimes \mathbf{I} \tag{Rows are i.i.d, same as before} \\
    \zeta\left(\mathbf{\Psi}_\rows, \mathbf{\Psi}_\cols\right) &= \mathbf{I} \otimes \mathbf{\Psi}_\rows \tag{Columns are i.i.d} \\
    \zeta\left(\mathbf{\Psi}_\rows, \mathbf{\Psi}_\cols\right) &= \mathbf{\Psi}_\cols \otimes \mathbf{\Psi}_\rows \tag{`Strong product' assumption} \\
    \zeta\left(\mathbf{\Psi}_\rows, \mathbf{\Psi}_\cols\right) &= \mathbf{\Psi}_\cols \otimes \mathbf{I} + \mathbf{I} \otimes \mathbf{\Psi}_\rows \tag{`Cartesian product' assumption} \\
    \zeta\left(\mathbf{\Psi}_\rows, \mathbf{\Psi}_\cols\right) &= \left(\mathbf{\Psi}_\cols \otimes \mathbf{I} + \mathbf{I} \otimes \mathbf{\Psi}_\rows\right)^2 \tag{Sylvester generative model}
\end{align*}

There are also models which break $\mathbf{\Omega}$ up into three axes, distinguishing between different \textit{types} of row dependencies \parencite{andrew_strong_2025}.  The strong product assumption is the oldest \parencite{dawid_matrix-variate_1981}, although it is often misreported as the `tensor product' assumption.  The Cartesian product assumption was introduced by \textcite{kalaitzis_bigraphical_2013} and the Sylvester generative model by \textcite{wang_sylvester_2020}.  For scRNA-seq data, the assumptions above correspond to the following statements:

\begin{assumption*}[Strong product assumption]
    The strength of the dependency between the expression of a gene $g_1$ in a cell $c_1$ and the expression of a gene $g_2$ in a cell $c_2$ is equal to the product of the dependencies between $c_1$ and $c_2$, and $g_1$ and $g_2$.

    In other words, genes depend on the expression of other genes within the cell and genes in other cells.
\end{assumption*}

\begin{assumption*}[Cartesian product assumption]
    A gene $g_1$ may only be conditionally dependent ($\sim$) on other genes $g_2$ within the same cell $c_1$, or the same gene in other cells $c_2$: $(c_1, g_1) \sim (c_1, g_2)$, $(c_1, g_1) \sim (c_2, g_1)$, and $(c_1, g_1) \nsim (c_2, g_2)$.  Just because they \textit{may} be dependent does not mean they must be.

    In other words, genes can only directly affect the expression of other genes within the same cell (or the same gene in different cells).
\end{assumption*}

\begin{assumption*}[Sylvester generative model]
    \[
    \mathbf{\Psi}_\rows \mathbf{D} + \mathbf{D} \mathbf{\Psi}_\cols \sim \mathcal{N}(\mathbf{0}, \mathbf{I})
    \]

    This model naturally arises when the dataset is generated from certain partial differential equations.
\end{assumption*}

For scRNA-seq data, the Cartesian product is the most natural.  While genes can interact with genes in other cells in certain scenarios, such interactions will be much weaker than within-cell interactions; the strong product assumption would give equal weight to within- and between-cell interactions.

The operation $\mathbf{A}\otimes\mathbf{I} + \mathbf{I}\otimes\mathbf{B}$ is denoted the `Kronecker sum': $\mathbf{A}\oplus \mathbf{B}$.  Many algorithms exist to fit Kronecker-sum-structured normal distributions, the most scalable being GmGM \cite{andrew_gmgm_2024}.  However, none fit such distributions under the scenario of multiplicative noise.

\section{Methodology}
\label{sec:med-magma-methodology}

To fit the model, we will perform maximum likelihood estimation (MLE) using expectation maximization (EM).  Let us rephrase our model for a dataset $\mathbf{x} = \mathrm{vec}\left[\mathbf{X}_{d_\rows \times d_\cols}\right]$ as follows:

\begin{align*}
    \mathbf{x} &= \left(\mathbf{r}_\cols \otimes \mathbf{r}_\rows\right) \circ \mathbf{z} \\ 
    \mathbf{z} &\sim \mathcal{N}\left(\mathbf{0}, \left(\mathbf{\Psi}_\cols \oplus \mathbf{\Psi}_\rows\right)^{-1}\right) \\
    \mathbf{r}_\rows,\mathbf{r}_\cols &\sim \text{any strictly positive distribution}
\end{align*}

The $\mathbf{r}$ variables are allowed to have dependencies among themselves, but they must be independent of the latent variable $\mathbf{z}$.  As a shorthand, we will let $\mathbf{\Omega} = \mathbf{\Psi}_\cols \oplus \mathbf{\Psi}_\rows$.  As further shorthand, we will implicitly let capital variables $\mathbf{V}$ and lowercase variables of the same name $\mathbf{v}$ follow the relationship $\mathbf{v} = \mathrm{vec}\left[\mathbf{V}\right]$.  Our final algorithm is summarized below.  In the following sections, we will derive the algorithm step-by-step.

\begin{algorithm*}[h!]
\caption{MED-MAGMA}
\label{alg:med-magma-algorithm}
\begin{algorithmic}
    \Require A $d_\rows \times d_\cols$ dataset $\mathbf{X}$ with vectorization $\mathbf{x}$
    \State $\mathbf{y} \gets f(\mathbf{x})$ \Comment{Remove noisy directions (Section \ref{sec:med-magma-deriving-density})}
    \State $\mathbf{\Psi}_\rows \gets \mathbf{I}$ \Comment{Initialize estimates}
    \State $\mathbf{\Psi}_\cols \gets \mathbf{I}$
    \Repeat
        \State Find $\mathbf{z}^*$ by flip-flopping axes while solving Equation \ref{eq:med-magma-flip-flop} \Comment{Section \ref{sec:med-magma-finding-z-star}}
        \State Approximate $\mathbf{S}_\rows \text{ and } \mathbf{S}_\cols$ using $\mathbf{\Psi}_\rows, \mathbf{\Psi}_\cols, \text{ and }\mathbf{z}^*$ with Equation \ref{eq:med-magma-pseudo-sufficient} \Comment{Section \ref{sec:med-magma-pseudo-sufficient}}
        \State $\left(\mathbf{\Psi}_\rows, \mathbf{\Psi}_\cols\right) \gets \mathrm{GmGM}\left(\mathbf{S}_\rows, \mathbf{S}_\cols\right)$ \Comment{Section \ref{sec:med-magma-performing-em}}
    \Until{$\mathbf{\Psi}_\rows \text{ and }\mathbf{\Psi}_\cols$ converge}
    \State \hspace{-8pt}\textbf{Output:} $\mathbf{\Psi}_\rows \text{ and } \mathbf{\Psi}_\cols$, which encode the conditional dependencies of the data.
\end{algorithmic}
\end{algorithm*}

\subsection{Handling non-identifiability}
\label{sec:med-magma-non-identifiability}

Our model generalizes the class of elliptical distributions, and thus cannot discriminate between $\mathbf{\Omega}$ and $r\mathbf{\Omega}$ for $r > 0$.  This is not a critical issue - we are only interested in the sparsity structure of $\mathbf{\Omega}$, not its scale, as it is the sparsity structure that encodes the conditional dependencies.  However, it must be addressed to ensure proper behavior of our algorithm.  Observe the following, for some measure of magnitude $\lVert\cdot\rVert$:

\begin{align*}
    \mathbf{x} = \left(\mathbf{r}_\cols \otimes \mathbf{r}_\rows\right) \circ \mathbf{z} = r\left(\mathbf{r}_\cols'\otimes\mathbf{r}_\rows'\right)\circ \mathbf{z} \text{ where } \lVert\mathbf{r}_\cols' \otimes \mathbf{r}_\rows'\rVert = 1
\end{align*}

This is why the non-identifiablilty arises: due to being free to choose $r$, $\mathbf{x}$ and $r\mathbf{x}$ are indistinguishable by the model, and hence so are $\mathbf{\Omega}$ and $r\mathbf{\Omega}$.  As we don't care about this difference, we can assume w.l.o.g. that $r=1$; in other words, w.l.o.g. $\lVert\mathbf{r}_\cols \otimes \mathbf{r}_\rows\rVert = 1$.  It will be convenient later if we choose $\lVert\cdot\rVert$ to be a geometric mean constraint: $\left(\prod_{i} r_{i}^\rows\right)^{\frac{1}{d_\rows}}\left(\prod_j r_j^\cols\right)^{\frac{1}{d_\cols}} = 1$.

We also have to address the non-identifiability of Kronecker products: $\mathbf{a} \otimes \mathbf{b} = c\mathbf{a} \otimes \frac{1}{c}\mathbf{b}$.  This non-identifiability is an artifact of the parameterization, and does not affect the extent to which $\mathbf{\Omega}$ can be learned.  We can simply modify our constraint to get rid of this artifact:  $\prod_{i} r_{i}^\rows = 1 = \prod_j r_j^\cols$.  Thus, we aim to fit the following:

\begin{align*}
    \mathbf{x} &= \left(\mathbf{r}_\cols \otimes \mathbf{r}_\rows\right) \circ \mathbf{z} \\ 
    \mathbf{z} &\sim \mathcal{N}\left(\mathbf{0}, \left(\mathbf{\Psi}_\cols \oplus \mathbf{\Psi}_\rows\right)^{-1}\right) \\
    \mathbf{r}_\rows,\mathbf{r}_\cols &\sim \text{any strictly positive distribution} \\
    \text{such that } & \prod_i r_i^\rows = 1 = \prod_j r_j^\cols
\end{align*}

\subsection{Removing the effect of the noise}
\label{sec:med-magma-deriving-density}

This model is a mixture of distributions (the normal distribution and any arbitrary strictly positive distribution), and thus does not directly yield a likelihood to estimate.  In fact, no likelihood will be available as long as we have such a general assumption on the $\mathbf{r}$ variables.  What we need is a function $f$ such that $f(\mathbf{x}) = f(\mathbf{y}) \iff \mathbf{x} = \left(\mathbf{r}_\cols \otimes \mathbf{r}_\rows\right) \circ \mathbf{y}$.  Such a function throws out all the information $\mathbf{x}$ provides that has become corrupted by the noise, leaving only the information upon which it is safe to do inference; we can then deduce a likelihood over $f$'s image.

It is worth thinking about the class of $d$-dimensional elliptical distributions again, in which we have $\mathbf{x} = r \mathbf{z}$.  The arbitrariness of ${r}$ ensures we can never know the magnitude of $\mathbf{z}$, but we can know its orientation.  In particular, $\frac{\mathbf{x}}{\lVert\mathbf{x}\rVert} = \frac{r\mathbf{z}}{\lVert r\mathbf{z}\rVert} = \frac{\mathbf{z}}{\lVert \mathbf{z}\rVert}$; $f(\mathbf{y}) = \frac{\mathbf{y}}{\lVert \mathbf{y} \rVert}$ maps the data to a circle.  The probability distribution at any point $\mathbf{p}$ along this circle is obtained by integrating over every point $\mathbf{y}$ such that $f(\mathbf{y})=f(\mathbf{p})$.  This space of points, denoted generically $\mathbb{F}$ (and denoted $f^{-1}(\mathbf{p})$ when considering a specific circle point $\mathbf{p}$), is called the `fiber space', and the remaining space $\mathbb{Q}$ (the circle) is the `quotient space'; $\mathbb{R}^d \cong \mathbb{F} \times \mathbb{Q}$.  We obtain the likelihood over $\mathbb{Q}$ by integrating out $\mathbb{F}$:

\begin{align*}
    \mathrm{pdf}\left(\mathbf{x}|\mathbf{\Omega}\right) &= \int_{\mathbf{y}\in f^{-1}\left(\mathbf{x}\right)} \mathrm{pdf}_\mathcal{N}\left(\mathbf{y}|\mathbf{\Omega}\right)d\mathbf{y}
\end{align*}


For the standard class of elliptical distributions, this integral is actually tractable, yielding the angular Gaussian distribution, whose MLE turns out to be Tyler's M estimator \parencite{tyler_distribution-free_1987}.  
For our distribution, the scenario is significantly more complicated, as normalizing to the unit sphere does not remove the noise.  Yet the procedure stays the same: we will produce a function $f$ that throws away unreliable information while preserving everything else, and then integrate out the fibers induced by $f$ to produce a formula for the likelihood over $\mathbb{Q}$.





\begin{align}
    f(\mathbf{X})_{ij} 
    &= \frac{x_{ij}\left|\prod_{i'j'} x_{i'j'}\right|^{\frac{1}{d_\rows d_\cols}}}{\Bigl|\prod_{j'} x_{ij'}\Bigr|^\frac{1}{d_\cols}\Bigl|\prod_{i'} x_{i'j}\Bigr|^\frac{1}{d_\rows}} \label{eq:med-magma-fiber}
\end{align}

\begin{proposition*}[Goodness of $f$]
    The function $f$ described in Equation \ref{eq:med-magma-fiber} has the following property: $f(\mathbf{X}) = f(\mathbf{Y})$ if and only if one can find strictly positive $\mathbf{r}_\rows,\mathbf{r}_\cols$ such that we have $\mathrm{vec}\left[\mathbf{X}\right] = \left(\mathbf{r}_\cols \otimes \mathbf{r}_\rows\right) \circ \mathrm{vec}\left[\mathbf{Y}\right]$.
\end{proposition*}
\begin{proof}
    See Appendix \ref{apx:med-magma-f-corrupted}.
\end{proof}

\paragraph{Handling zeros} Many datasets, such as most scRNA-seq datasets, contain zeros.  Equation \ref{eq:med-magma-fiber} as written is undefined for such zeros.  We present a zero-handling modification of $f$ in Appendix \ref{apx:med-magma-f-zeros}.  From here onwards, we will assume that $f$ works with zeros.

\subsection{Performing expectation maximization}
\label{sec:med-magma-performing-em}

Now that we have $f$, we can express the likelihood of our distribution as:

\begin{align*}
\mathrm{pdf}\left(\mathbf{x}|\mathbf{\Omega}\right) &= \int_{\mathbf{z}\in f^{-1}\left(\mathbf{x}\right)} \mathrm{pdf}_\mathcal{N}\left(\mathbf{z}|\mathbf{\Omega}\right)d\mathbf{z} 
\hspace{5pt}\propto\hspace{5pt} \sqrt{\mathrm{det}\text{ }\mathbf{\Omega}}\int_{\mathbf{z}\in f^{-1}\left(\mathbf{x}\right)}  e^{-\frac{1}{2}\mathbf{z}^T\mathbf{\Omega}\mathbf{z}}d\mathbf{z}
\end{align*}

Sadly, this integral is likely intractable due to the complicated structure of $f^{-1}$.  This motivates the use of the EM algorithm, frequently used in settings with latent variables such as this.  Letting $\mathcal{K}_{++}^\oplus$ be the space of positive definite Kronecker-sum-separable matrices (i.e. the space of valid $\mathbf{\Omega}$), the EM algorithm is as follows:

\begin{align*}
    \text{let } \mathbf{y} &= f(\mathbf{x}) = f(\mathbf{z}) \text{ be our denoised observed variable}\\
    \text{let } Q\left(\mathbf{\Omega}\Bigl|\mathbf{\Omega}_{(t)}\right) &= \mathbb{E}_{\mathbf{z}\sim\mathcal{N}\left(\mathbf{0},\mathbf{\Omega}_{(t)}^{-1}\right)\text{ and } \mathbf{z}\in f^{-1}(\mathbf{y})}\left[\log \mathrm{pdf}_\mathcal{N}(\mathbf{z} | \mathbf{\Omega})\right]\\
    \mathbf{\Omega}_{(t+1)} &= \underset{\mathbf{\Omega}\in\mathcal{K}_{++}^\oplus}{\mathrm{argmax}}\hspace{3pt} Q\left(\mathbf{\Omega}\Bigl|\mathbf{\Omega}_{(t)}\right) \tag{iterate until convergence}
\end{align*}

Thus we have reduced our problem to the (admittedly intimidating) problem of maximizing $Q$.

\begin{proposition*}[Maximizing $Q$]
    Optimizing $Q$ corresponds to \textbf{the same optimization problem as GmGM}.
    
    GmGM uses the sufficient statistics $\mathbf{X}\mathbf{X}^T$ and $\mathbf{X}\mathbf{X}^T$.  To maximize $Q$, we replace these with:
    
    \[\mathbf{S}^\rows = \mathbb{E}_{\begin{matrix}\scriptscriptstyle\mathbf{z}\sim\mathcal{N}\left(\mathbf{0},\mathbf{\Omega}_{(t)}^{-1}\right)\\\scriptscriptstyle\mathbf{z}\in f^{-1}(\mathbf{y})\end{matrix}}\left[\mathbf{Z}\mathbf{Z}^T\right] \text{ and } \mathbf{S}^\cols = \mathbb{E}_{\begin{matrix}\scriptscriptstyle\mathbf{z}\sim\mathcal{N}\left(\mathbf{0},\mathbf{\Omega}_{(t)}^{-1}\right)\\\scriptscriptstyle\mathbf{z}\in f^{-1}(\mathbf{y})\end{matrix}}\left[\mathbf{Z}^T\mathbf{Z}\right]\]

    More generally, any method to fit unregularized Kronecker-sum-structured models \textbf{without noise} that operates on the sufficient statistics can be used in place of GmGM.
\end{proposition*}
\begin{proof}
    The proof is given in Appendix \ref{apx:med-magma-q-gmgm}.
\end{proof}

This is excellent: GmGM is a fast algorithm.  All that remains is to evaluate the expectations.

\subsection{Evaluating $\mathbf{S}^\rows \text{ and } \mathbf{S}^\cols$ with Laplace's approximation}
\label{sec:med-magma-pseudo-sufficient}

Let $g(\mathbf{z}) \text{ be either } \mathbf{Z}\mathbf{Z}^T \text{ or } \mathbf{Z}^T\mathbf{Z}$; the logic behind evaluating the expectation is the same.  We can express the expectation as follows:

\begin{align}
    \mathbb{E}_{\begin{matrix}\scriptscriptstyle\mathbf{z}\sim\mathcal{N}\left(\mathbf{0},\mathbf{\Omega}_{(t)}^{-1}\right)\\\scriptscriptstyle\mathbf{z}\in f^{-1}(\mathbf{y})\end{matrix}}\left[g(\mathbf{z})\right] 
    &= \frac{\int_{\mathbf{z}\in f^{-1}\left(\mathbf{y}\right)} g(\mathbf{z})e^{-\frac{1}{2}\mathbf{z}^T\mathbf{\Omega}_{(t)}\mathbf{z}} d\mathbf{z}}{\int_{\mathbf{z}\in f^{-1}\left(\mathbf{y}\right)} e^{-\frac{1}{2}\mathbf{z}^T\mathbf{\Omega}_{(t)}\mathbf{z}} d\mathbf{z}} \label{eq:med-magma-pseudo-sufficient}
\end{align}

It may appear that we have swapped one hard problem for another - we were trying to avoid a nasty integral-over-fiber in Section \ref{sec:med-magma-performing-em}, yet now we have a ratio of them!  However, this ratio can be well-approximated by Laplace's method 
\parencite{tierney_accurate_1986}.  The intuition behind Laplace's method is that the integral is dominated by the exponential term, whose maximum occurs at $\mathbf{z}^{*}$ (in other words, the minimum of $\mathbf{z}^T\mathbf{\Omega}_{(t)}\mathbf{z}$ occurs at $\mathbf{z}^{*}$).  By focusing on local approximations about $\mathbf{z}^{*}$, we can well-approximate the integral in the region where it has the most mass.  Letting $\mathbf{P}$ be the matrix that projects onto the tangent space about $\mathbf{z}^*$ in the fiber, Laplace's method tells us that:

\begin{align*}
    \frac{\int_{\mathbf{z}\in f^{-1}\left(\mathbf{y}\right)} g(\mathbf{z})e^{-\frac{1}{2}\mathbf{z}^T\mathbf{\Omega}_{(t)}\mathbf{z}} d\mathbf{z}}{\int_{\mathbf{z}\in f^{-1}\left(\mathbf{y}\right)} e^{-\frac{1}{2}\mathbf{z}^T\mathbf{\Omega}_{(t)}\mathbf{z}} d\mathbf{z}} &\approx g\left(\mathbf{z}^{*}\right) + \frac{1}{2}\mathrm{tr}\left[\left(\mathbf{P}\mathbf{\Omega}_{(t)}\mathbf{P}^T\right)^{-1}\mathbf{P}\left(\nabla^2 g(\mathbf{z}^*)\right)\mathbf{P}^T\right]
\end{align*}

The trace term above looks intimidating, but in practice it can be evaluated efficiently; we show how in Appendix \ref{apx:med-magma-trace-term}.  Thus, all that remains is to find the point $\mathbf{z}^{*} \in f^{-1}(\mathbf{y})$ that minimizes $\mathbf{z}^T\mathbf{\Omega}_{(t)}\mathbf{z}$.

\subsection{Finding $\mathbf{z}^{*}$ with quadratic programming}
\label{sec:med-magma-finding-z-star}

In order to approximate $\mathbf{S}^\rows \text{ and } \mathbf{S}^\cols$ statistics, we need to solve the following optimization problem:

\begin{align*}
    \mathbf{z}^* &= \underset{\mathbf{z}\in f^{-1}(\mathbf{y})}{\mathrm{argmin}}\hspace{2pt} \mathbf{z}^T\mathbf{\Omega}_{(t)}\mathbf{z} \\
    &= \underset{\begin{matrix}\mathbf{r}_\rows, \mathbf{r}_\cols > 0 \\ \prod_i r^{\rows}_i = 1 = \prod_j r^{\cols}_j \end{matrix}}{\mathrm{argmin}}\hspace{2pt} \sum_{ik} r^\rows_i \left(\sum_{j\ell}r^\cols_j Y_{ij} \Omega^{(t)}_{(ij),(k\ell)} r^\cols_\ell Y_{k\ell}\right) r^\rows_k
\end{align*}

This is equivalent to minimizing an order-4 polynomial with constraints - not a trivial task at all.  However, our problem naturally lends itself to a flip-flop algorithm over two quadratic sub-problems; the terms in the parentheses may be aggregated into a $d_\rows \times d_\rows$ matrix $\Tilde{\mathbf{\Omega}}$ not depending on $\mathbf{r}_\rows$.  This is a quadratic form in $\mathbf{r}_\rows$, and so for fixed $\mathbf{r}_\cols$ we can reduce our problem:



\begin{align}
    \mathrm{argmin}\hspace{2pt} \mathbf{r}_\rows^T \Tilde{\mathbf{\Omega}} \mathbf{r}_\rows \text{ such that } \mathbf{r}_\rows > 0 \text{ and } \prod_i r^{\rows}_i = 1\label{eq:med-magma-flip-flop}
\end{align}

This is still non-convex\footnote{We did experiment with a sum-to-one constraint, which would be convex; it performed well on synthetic data but not real data, due to the constraint not respecting the geometry of the problem nor the structure of $f$.} due to the nonlinear multiply-to-one constraint, but in practice is efficiently optimizable.  An analogous formulation holds for $\mathbf{r}_\cols$; we can flip-flop between them to find the optimal $\mathbf{r}_\rows, \mathbf{r}_\cols$, and hence $\mathbf{z}^*$.




\section{Results}
\label{sec:med-magma-results}

In this section, we evaluate our method, MED-MAGMA, in three ways: by comparing it to GmGM on synthetic data (Section \ref{sec:med-magma-synthetic}), comparing it to GmGM systematically across all public datasets of a certain type (Section \ref{sec:med-magma-systemic}), and finally by exploring a specific example (Section \ref{sec:med-magma-specific}).  The code to run these experiments is available on GitHub (\href{https://github.com/MED-MAGMA-NeurIPS2026/MED-MAGMA-NeurIPS2026}{https://github.com/MED-MAGMA-NeurIPS2026/MED-MAGMA-NeurIPS2026}) and our model is available on PyPI (\texttt{pip install MED-MAGMA}).

\subsection{Synthetic validation}
\label{sec:med-magma-synthetic}

\begin{figure}[h!]
    \centering
    \includegraphics[width=0.9\linewidth]{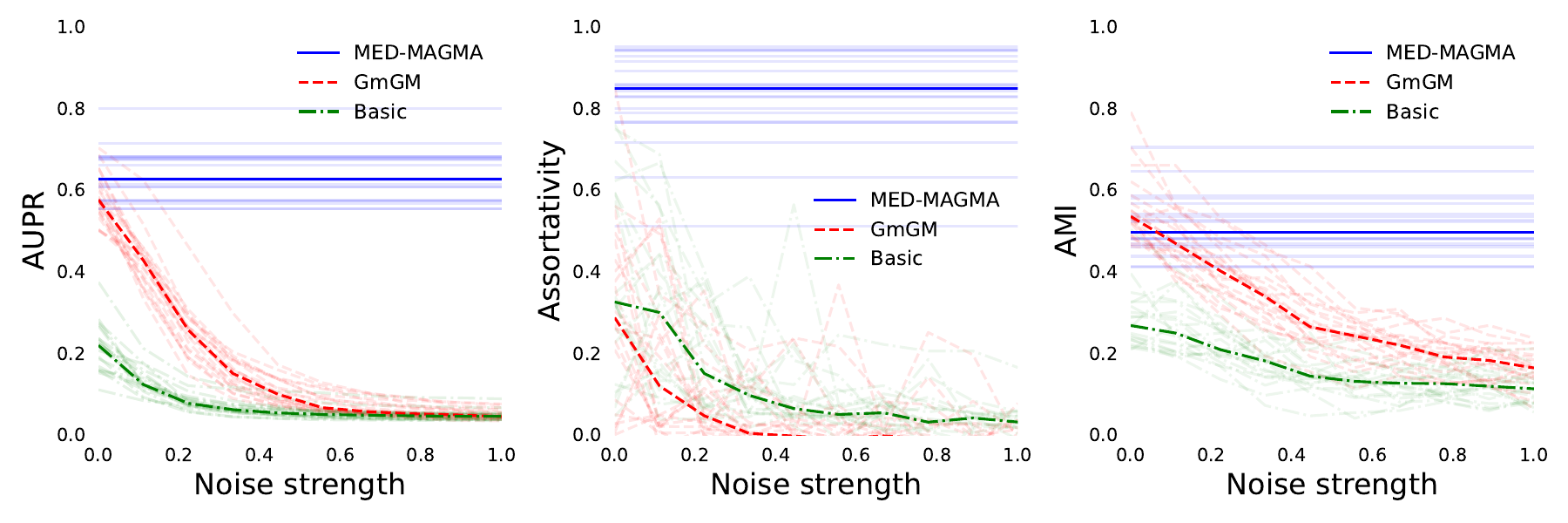}
    \vspace{-14pt}
    \caption{(Both) Results across 20 synthetic Kronecker-sum-structured Gaussian $100 \times 150$ datasets corrupted by varying levels of multiplicative noise.  Median result is highlighted.  As noise strength varies, we plot the area under PR curves (left), assortativity (center), and AMI (right).}
    \label{fig:med-magma-pr-aupr}
\end{figure}

To validate our method on synthetic data, we first generated ground truth graphs from the Barabasi-Albert distribution, and converted these to positive definite matrices $\mathbf{\Psi}^\rows$ and $\mathbf{\Psi^\cols}$.  We then generated 20 samples from $\mathcal{N}\left(\mathbf{0}, \left(\mathbf{\Psi}_\cols \oplus \mathbf{\Psi}_\rows\right)^{-1}\right)$.  To corrupt the samples with multiplicative noise, we generated each $\mathbf{r}_{(\ell)}$ variable as chi-squared (df=1) and raised them to the power $\alpha \in [0, 1]$; $\alpha$ represents the noise strength ($\alpha=0$ yields no noise).  For two of our metrics (assortativity and AMI), comparison requires a ground truth clustering - we used Leiden clustering \parencite{traag_louvain_2019} with resolution 1 on the ground truth graphs to create this.  For AMI, we also need a clustering on the learned graphs - we measured AMI using Leiden clustering at various resolutions (from 0.02 to 2) and recorded the result with the best AMI.  The generated datasets presented here are $100\times 150$ matrices; in Appendix \ref{apx:med-magma-synthetic-bonus} we show results on $200\times 200$ matrices.

We compare MED-MAGMA and GmGM on this synthetic dataset in Figure \ref{fig:med-magma-pr-aupr}, using the metrics of precision and recall, assortativity, and AMI.  Assortativity measures the tendency of vertices of the same type to be connected together, whereas AMI measures the agreement between two clusterings.  Intuitively, these three metrics measure the ability of methods to capture structure at different scales, with precision and recall measuring the recovery of specific edges, AMI measuring the recovery of clusterings, and assortativity sitting in between.  We also compare the results against a baseline single-axis precision-matrix estimator (the MLE, which for a dataset $\mathbf{X}$ is just $\left(\mathbf{X}\mathbf{X}^T\right)^{-1}$).  We can see that MED-MAGMA is unaffected by noise strength, whereas GmGM and our baseline are quite severely affected.

\subsection{Systematic validation}
\label{sec:med-magma-systemic}

To validate MED-MAGMA on real data, we systematically downloaded every dataset in the Single Cell Expression Atlas less than a thousand cells which contained a categorical variable and referenced a source publication.  There were 30 such datasets in total; we list these in Appendix \ref{apx:med-magma-datasets}.  By evaluating our methodology in this manner, we ensure that we evaluate over a representative sample of all scRNA-seq datasets.

The datasets received minimal preprocessing, described in Appendix \ref{apx:med-magma-preprocessing}.  Each dataset contained a categorical variable, typically `cell type'.  For comparison, we used both assortativity and adjusted mutual information \parencite{maan_characterizing_2024}.  We thresholded the learned graphs to keep only the top $k$ edges per vertex, with $k$ varying from 1 to 40.  We used Leiden clustering, varying the resolution parameter from $0.02$ to $2.0$ to obtain clusterings of various sizes - we only report the best AMI across all clusterings and thresholds.  We compared against both standard GmGM, GmGM equipped with the nonparanormal skeptic, and our single-axis baseline; Figure \ref{fig:med-magma-systematic} demonstrates that assortativity and AMI are higher in MED-MAGMA than the other methods, with or without the nonparanormal skeptic.

\begin{figure}[h!]
    \centering
    \includegraphics[width=0.8\linewidth]{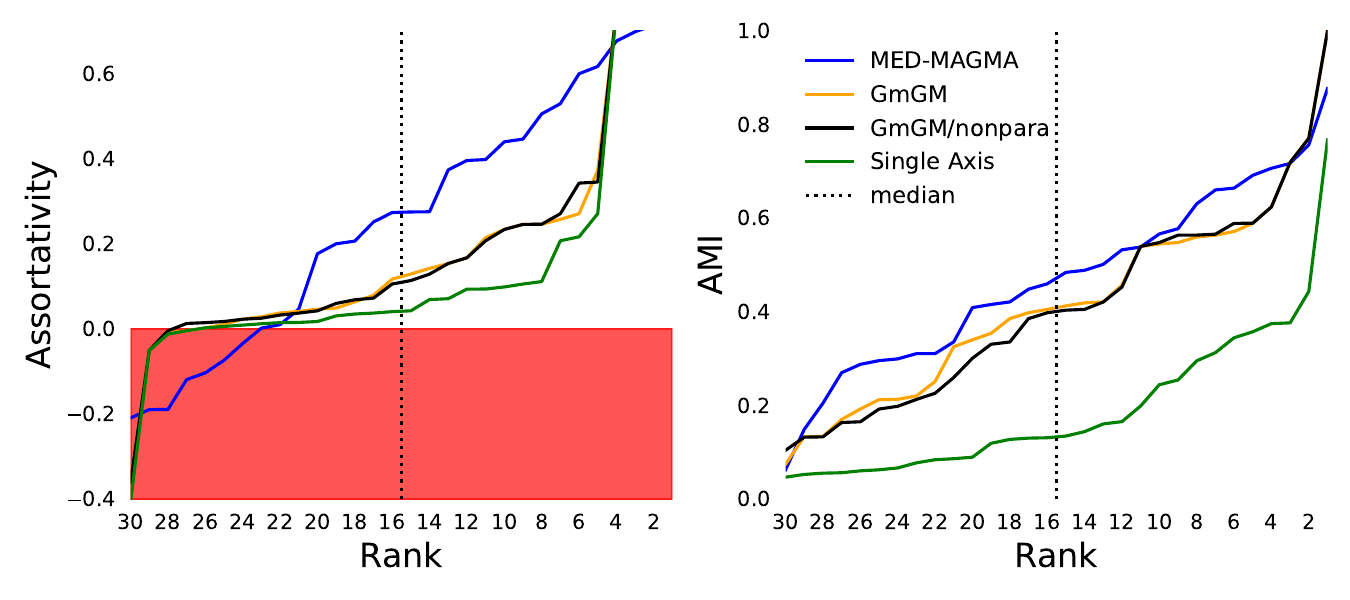}
    \vspace{-14pt}
    \caption{(Both) The datasets are ordered by increasing performance along the x-axis.  GmGM/nonpara refers to GmGM equipped with the nonparanormal skeptic.  The median performance is highlighted.  (Left) Assortativities below 0 represent a statistical tendency for cells of \textit{different} categories to connect, which indicates the network does capture information about the dataset, but it is harder to interpret.}
    \label{fig:med-magma-systematic}
\end{figure}

\subsection{Specific evaluation}
\label{sec:med-magma-specific}

Both prior sections serve to show that MED-MAGMA yields better results than GmGM.  Thus, models that take into account multiplicative noise perform better on scRNA-seq data.  In this section, we briefly pick one of the datasets from earlier and show that the cell and gene networks are robust, and that the gene network learned is also reasonable.  In particular, we pick the `median dataset' when ranking datasets by assortativity.  By picking the median dataset, it gives us a generic view into what we can expect the average performance of our methodology to be.

The chosen dataseet, E-MTAB-7249, contains 185 cells from a patient with high-grade serous ovarian carcinoma (HGSOC) \parencite{nelson_living_2020}.  The source paper investigated the biology of HGSOC through several modalities of data from multiple patients.  Here, we are using the portion of this dataset available on the Single Cell Expression Atlas, comprised only of the cells of patient `OCM.38a'.  The data from the atlas had 7 clusters representing cell subtypes: SA-a, SA-b, SB, TA-a, TA-b, TB, and a final small `N/A' cluster of unlabeled cells.  From the source paper, it is clear that these are subtypes of the following cell types: stromal (SA-a, SA-b, SB) and neoplastic (TA-a, TA-b, TB).

\begin{figure}[h!]
    \centering
    \includegraphics[width=\linewidth]{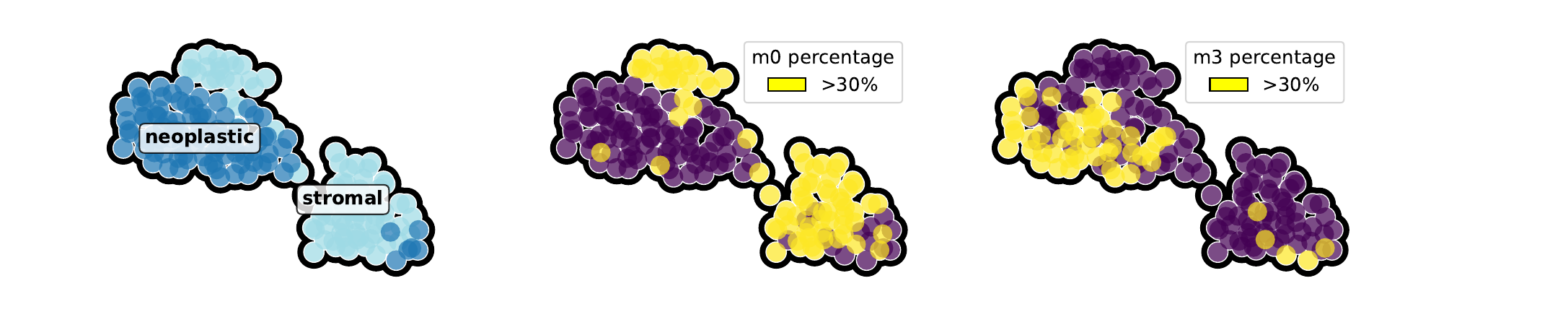}
    \vspace{-18pt}
    \caption{(All) A UMAP plot of E-MTAB-7249 cells based on gene expression.  (Left) Cell type. (Middle, Right) The percentage of gene expression from genes in modules m0 and m3, respectively.}
    \label{fig:med-magma-umap}
\end{figure}

We clustered the learned gene network (Leiden, resolution=1), yielding 8 `modules'.  Figure \ref{fig:med-magma-umap} shows that two of our gene modules strongly correspond to stromal and neoplastic cells.  To quantify the relation between our gene modules and the cell subtypes, we performed a one-way ANOVA testing whether the percentage expression of each gene module differed across cell subtypes.  All modules showed significant variation between clusters ($p=2\times 10^{-29} \text{ to } 1\times 10^{-2}$) except one ($p=0.11$), indicating that MED-MAGMA learned gene modules capturing processes associated with specific cell subtypes.

For one final test, we evaluated the robustness of our cell and gene networks using Robin \parencite{policastro_robustness_2021}.  Robin measures the tendency for clusterings on a network to change after an increasing proportion of the edges are randomly perturbed: networks whose clusterings change significantly after a small clustering are not robust.  Robin outputs Bayes factors; our cell network has factor of 108 and our gene network has a factor of 59, indicating very strong evidence that clusterings on MED-MAGMA's networks are robust to perturbations.

\section{Discussion and limitations}
\label{sec:med-magma-discussion}

We have developed a model that allows multi-axis estimation under multiplicative noise, and an algorithm to fit the model.  Such a model is important; many datasets, such as scRNA-seq datasets, are affected by this type of noise.  It also significantly increases the flexibility of multi-axis models; we did not make a Gaussian copula assumption, allowing our model to be suitable for datasets that contain more complicated modes of dependency (such as tail dependence).  We demonstrated that our methodology outperforms prior work on synthetic data as well as through a systematic evaluation of all public scRNA-seq datasets below a certain size on a public database.

Despite the promising results, our model does have limitations.  Currently, our method can handle matrices with thousands of rows/columns, but would struggle with tens of thousands.  Both the approximation of $\left(\mathbf{S}_\rows, \mathbf{S}_\cols\right)$ and the GmGM sub-problem dominate the runtime, so improvements in either of those would be significant.  Additionally, elliptical-style distributions such as ours (despite their generality) still make some tangible assumptions about the dependencies in the dataset - namely that dependencies are linear and radially symmetric.  It would be a worthwhile future direction to further weaken the assumptions made here.


\begin{ack}
    Bailey Andrew is supported by the UKRI Engineering and Physical Sciences Research Council (EPSRC) [EP/S024336/1].  David Westhead is supported in part by the National Institute for Health and Care Research (NIHR) Leeds Biomedical Research Centre (BRC) (NIHR203331). The views expressed are those of the author(s) and not necessarily those of the NHS, the NIHR or the Department of Health and Social Care.
\end{ack}


\printbibliography

\appendix

\section{Synthetic results for $200\times 200$ datasets}
\label{apx:med-magma-synthetic-bonus}

\begin{figure}[h!]
    \centering
    \includegraphics[width=0.9\linewidth]{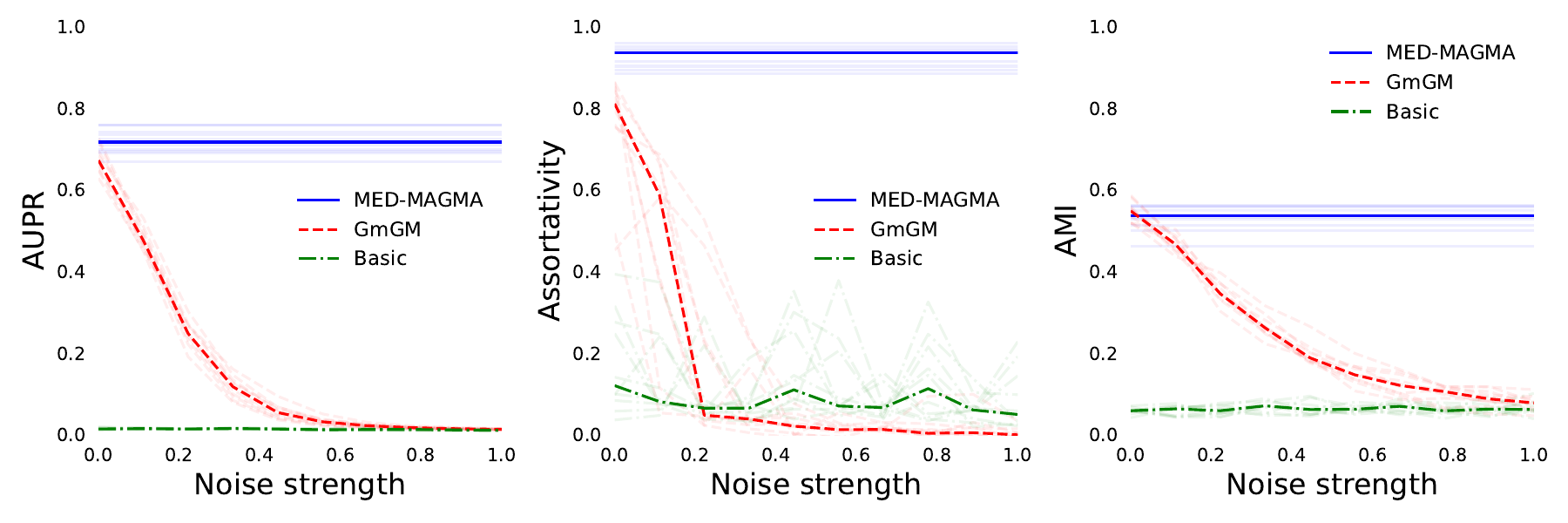}
    \vspace{-14pt}
    \caption{(Both) Results across 20 synthetic Kronecker-sum-structured Gaussian $200 \times 200$ datasets corrupted by varying levels of multiplicative noise.  Median result is highlighted.  As noise strength varies, we plot the area under PR curves (left), assortativity (center), and AMI (right).}
    \label{fig:med-magma-pr-aupr-appendix}
\end{figure}

We give the results for $200\times 200$ synthetic datasets, generated the same way as in Section \ref{sec:med-magma-synthetic}, in Figure \ref{fig:med-magma-pr-aupr-appendix}.  This was omitted from the main paper for space reasons; we include it here to demonstrate that our method works regardless of whether the data is lopsided or more square.  The basic single-axis baseline performs poorly on this test; this is because such methods work best when one has sufficiently more samples relative to features, whereas here we have equal samples/features.

\section{Proof that $f$ removes only the corrupted directions}
\label{apx:med-magma-f-corrupted}


In this section, we will demonstrate that our definition of $f$ (Equation \ref{eq:med-magma-fiber}) removes all directions corrupted by noise from our data, and nothing else.

\begin{proposition*}[Goodness of $f$]
    The function $f$ described in Equation \ref{eq:med-magma-fiber} has the following property: $f(\mathbf{X}) = f(\mathbf{Y})$ if and only if one can find strictly positive $\mathbf{r}_\rows,\mathbf{r}_\cols$ such that we have $\mathrm{vec}\left[\mathbf{X}\right] = \left(\mathbf{r}_\cols \otimes \mathbf{r}_\rows\right) \circ \mathrm{vec}\left[\mathbf{Y}\right]$.
\end{proposition*}
\begin{proof}

Recall the following direction of $f$:

\begin{align}
    f(\mathbf{X})_{ij} 
    &= \frac{x_{ij}\left|\prod_{i'j'} x_{i'j'}\right|^{\frac{1}{d_\rows d_\cols}}}{\Bigl|\prod_{j'} x_{ij'}\Bigr|^\frac{1}{d_\cols}\Bigl|\prod_{i'} x_{i'j}\Bigr|^\frac{1}{d_\rows}} \tag{Equation \ref{eq:med-magma-fiber}}
\end{align}

Note that this has the following equivalent, more complicated formulation:

\begin{align*}
    f(\mathbf{X})_{ij} &= \mathrm{sgn}\left(x_{ij}\right) \mathrm{exp}(& \\
    &\hspace{10pt}\log \left|x_{ij}\right| + \frac{1}{d_\rows d_\cols}\sum_{i'j'}\log\left|x_{i'j'}\right| - \frac{1}{d_\rows}\sum_{i'}\log\left|x_{i'j}\right| - \frac{1}{d_\cols}\sum_{j'}\log\left|x_{ij'}\right|\\
    ) \\
\end{align*}

Letting $\mathbf{L} = \log^\circ \left|\mathbf{X}\right|$ be the elementwise logarithm of the absolute values of $\mathbf{X}$, we can see that the core of our function is ultimately just a linear function on $\mathbf{L}$, $\mathcal{L}$:

\begin{align*}
    \mathcal{L}(\mathbf{L})_{ij} &= \ell_{ij} + \mathrm{mean}(\mathbf{L}) - \mathrm{row\text{ }means}\left(\mathbf{L}\right)_i - \mathrm{col\text{ }means}\left(\mathbf{L}\right)_j
\end{align*}

This is a well-known function (double-centering of a matrix) and it is easy to express it in matrix form:

\begin{align*}
    \mathcal{L}(\mathbf{L}) &= \mathbf{L} + \frac{\mathbf{1}^T\mathbf{L}\mathbf{1}}{d_\cols d_\rows}\mathbf{1}\mathbf{1}^T - \left(\frac{1}{d_\cols}\mathbf{L}\mathbf{1}\right)\mathbf{1}^T - \mathbf{1}\left(\frac{1}{d_\rows}\mathbf{1}^T\mathbf{L}\right) \\
    &= \mathbf{L} - \mathbf{L}\left(\frac{1}{d_\cols}\mathbf{1}\mathbf{1}^T\right) - \left(\frac{1}{d_\rows}\mathbf{1}\mathbf{1}^T\right)\mathbf{L} + \left(\frac{1}{d_\rows}\mathbf{1}\mathbf{1}^T\right)\mathbf{L}\left(\frac{1}{d_\cols}\mathbf{1}\mathbf{1}^T\right) \\
    &= \mathbf{L}\left(\mathbf{I} - \frac{1}{d_\cols}\mathbf{1}\mathbf{1}^T\right) + \left(\frac{1}{d_\rows}\mathbf{1}\mathbf{1}^T\right)\mathbf{L}\left(\mathbf{I} - \frac{1}{d_\cols}\mathbf{1}\mathbf{1}^T\right) \\
    &= \left(\mathbf{I} - \frac{1}{d_\rows}\mathbf{1}\mathbf{1}^T\right)\mathbf{L}\left(\mathbf{I} - \frac{1}{d_\cols}\mathbf{1}\mathbf{1}^T\right)
\end{align*}

In other words:

\begin{align*}
    f(\mathbf{X}) &= \mathrm{sgn}\left(\mathbf{X}\right) \circ \exp^\circ\left(\left[\mathbf{I} - \frac{1}{d_\rows}\mathbf{1}\mathbf{1}^T\right]\log^\circ\left|\mathbf{X}\right|  \left[\mathbf{I} - \frac{1}{d_\cols}\mathbf{1}\mathbf{1}^T\right] \right)\\
    f(\mathbf{x})&= \mathrm{sgn}\left(\mathbf{x}\right) \circ \exp^\circ\left(\left(\left[\mathbf{I} - \frac{1}{d_\cols}\mathbf{1}\mathbf{1}^T\right]  \otimes\left[\mathbf{I} - \frac{1}{d_\rows}\mathbf{1}\mathbf{1}^T\right]\right)\log^\circ\left|\mathbf{x}\right| \right) \tag{where $\mathbf{x}$ = $\mathrm{vec}\left[\mathbf{X}\right]$}\\
\end{align*}

The benefit of expressing our function in this form is that we can fully describe what information is lost and kept by the eigenvectors of $\left[\mathbf{I} - \frac{1}{d_\cols}\mathbf{1}\mathbf{1}^T\right]  \otimes\left[\mathbf{I} - \frac{1}{d_\rows}\mathbf{1}\mathbf{1}^T\right]$.  This is the Kronecker product of two centering matrices - the $n\times n$ centering matrix has a single zero eigenvalue, with the remaining $n-1$ eigenvalues being 1.  The eigenvalues of a Kronecker product are all the products of the eigenvalues, giving $d_\rows + d_\cols - 1$ zero eigenvalues (the rest being one).  If we can show that our noise gets mapped to zero and spans a $\left(d_\rows + d_\cols - 1\right)$-dimensional subspace, we will have proven that our function works.

\begin{align*}
    \log^\circ | r\left(\mathbf{r}_\cols \otimes \mathbf{r}_\rows\right)\circ \mathbf{x}| &= \log^\circ \mathbf{r}_\cols \oplus \log^\circ \mathbf{r}_\rows + \log^\circ \left|\mathbf{x}\right| + \log(r)\mathbf{1} \tag{Noise is linear in logspace}
\end{align*}
\begin{align*}
    \left(\left[\mathbf{I} - \frac{1}{d_\cols}\mathbf{1}\mathbf{1}^T\right]  \otimes\left[\mathbf{I} - \frac{1}{d_\rows}\mathbf{1}\mathbf{1}^T\right]\right)&\left(\log^\circ \mathbf{r}_\cols \oplus \log^\circ \mathbf{r}_\rows + \log(r)\mathbf{1}\right) \\
    &= \left(\left[\mathbf{I} - \frac{1}{d_\cols}\mathbf{1}\mathbf{1}^T\right]\log^\circ \mathbf{r}_\cols\right) \otimes \left[\mathbf{I} - \frac{1}{d_\rows}\mathbf{1}\mathbf{1}^T\right] \\
    &+ \left[\mathbf{I} - \frac{1}{d_\cols}\mathbf{1}\mathbf{1}^T\right] \otimes \left(\left[\mathbf{I} - \frac{1}{d_\rows}\mathbf{1}\mathbf{1}^T\right]\log^\circ\mathbf{r}_\rows\right) \\
    &+ \log(r) \left[\mathbf{I} - \frac{1}{d_\cols}\mathbf{1}\mathbf{1}^T\right]\mathbf{1} \otimes \left[\mathbf{I} - \frac{1}{d_\rows}\mathbf{1}\mathbf{1}^T\right]\mathbf{1}
\end{align*}

Our constraint that $\prod_i \mathbf{r}^{(\ell)}_i = 1$ is equivalent to $\sum_i \log \mathbf{r}^{(\ell)}_i = 0$.  In other words, in logspace our vectors have mean zero, and are thus annihilated by the centering matrix.  The whole equation then becomes zero: our noise is indeed removed by our function.

Finally, note that we parameterize $r, \mathbf{r}_\rows\text{ and }\mathbf{r}_\cols$ with $1+d_\rows+d_\cols$ parameters.  However, our two geometric mean constraints further decrease the dimensionality; our noise is parameterized by $d_\rows + d_\cols - 1$ parameters.  This is exactly the size of the nullspace of our double-centering matrix, guaranteeing that only noise-corrupted directions are removed from $\mathbf{x}$.  This completes the proof

\end{proof}

\section{A version of $f$ that handles zero values}
\label{apx:med-magma-f-zeros}

\begin{figure}[h!]
    \centering
    \includegraphics[width=0.6\linewidth]{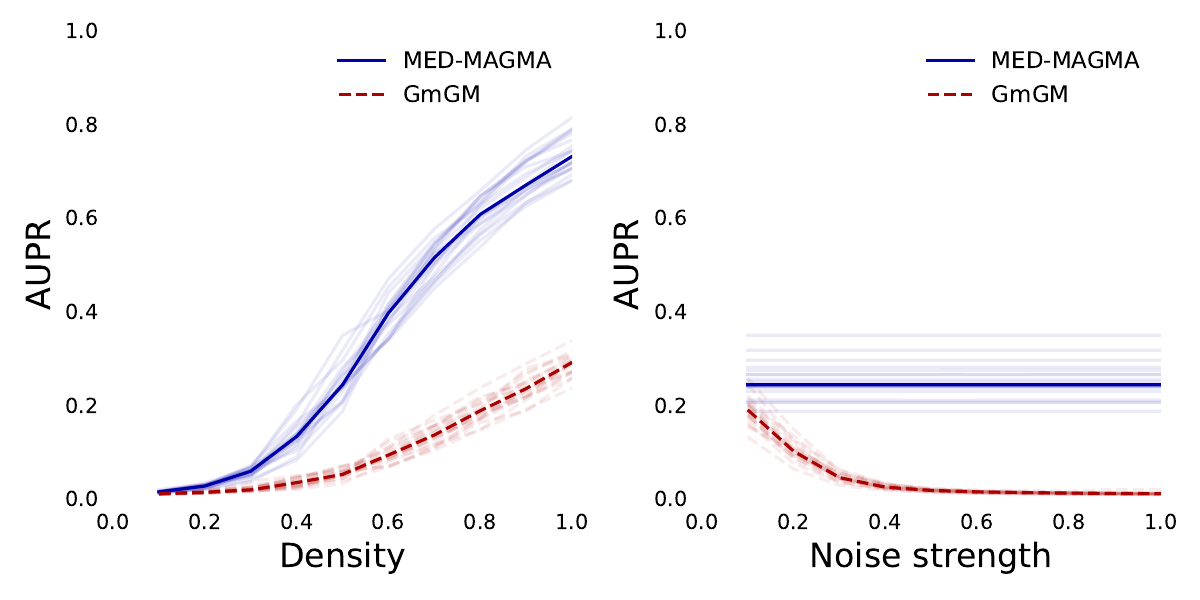}
    \vspace{-12pt}
    \caption{(Both) Area under PR curves across 20 synthetic $200\times 201$ Kronecker-sum-structured Gaussian datasets corrupted by varying levels of multiplicative noise and sparsity.  Median result is highlighted.  (Left) As percentage of nonzero entries vary (noise strength 0.2)  (Right) Area under PR curves as noise strength varies (density 0.5).}
    \label{fig:med-magma-sparsity}
\end{figure}

When the input $\mathbf{X}$ contains zeros, the function $f$ breaks - either one must divide by zero or take the logarithm of it.  However, zeros are unaffected by multiplicative noise; intuitively, they should not cause such a catastrophic failure.  We use the obvious modification (with Figure \ref{fig:med-magma-sparsity} demonstrating the performance): we will skip over the zero entries when doing products.  Let $\mathrm{nnz}$ be a function that counts the number of nonzero entries in a vector, and $\prod^{\neq 0}$ be a product over only nonzero terms; we have the following: 

\begin{align*}
    f(\mathbf{X})_{ij} 
    &= \frac{x_{ij}\left|\prod_{i'j'}^{\neq 0} x_{i'j'}\right|^{\frac{1}{\mathrm{nnz}\left(\mathbf{x}\right)}}}{\Bigl|\prod_{j'}^{\neq 0} x_{ij'}\Bigr|^\frac{1}{\mathrm{nnz}\left(\mathbf{x}_{\mathrm{row}\text{ }i}\right)}\Bigl|\prod_{i'}^{\neq 0} x_{i'j}\Bigr|^\frac{1}{\mathrm{nnz}\left(\mathbf{x}_{\mathrm{col}\text{ }j}\right)}}
\end{align*}


\section{Proof that optimizing $Q$ is equivalent to GmGM}
\label{apx:med-magma-q-gmgm}

Recall the problem of optimizing $Q$:

\begin{align*}
    \text{let } Q\left(\mathbf{\Omega}\Bigl|\mathbf{\Omega}_{(t)}\right) &= \mathbb{E}_{\mathbf{z}\sim\mathcal{N}\left(\mathbf{0},\mathbf{\Omega}_{(t)}^{-1}\right)\text{ and } \mathbf{z}\in f^{-1}(\mathbf{y})}\left[\log \mathrm{pdf}_\mathcal{N}(\mathbf{z} | \mathbf{\Omega})\right]\\
    \mathbf{\Omega}_{(t+1)} &= \underset{\mathbf{\Omega}\in\mathcal{K}_{++}^\oplus}{\mathrm{argmax}}\hspace{3pt} Q\left(\mathbf{\Omega}\Bigl|\mathbf{\Omega}_{(t)}\right) \tag{iterate until convergence}
\end{align*}

As claimed in the main paper, maximization of $Q$ requires solving the same optimization problem as GmGM.  Below, let $\mathbf{J}^{ij}$ be the matrix of zeros except for the $(i, j)$ entry being one.

\begin{definition*}[GmGM \parencite{andrew_gmgm_2024}]
    GmGM solves the convex minimization problem of minimizing the negative log likelihood, with gradients:
    
    \[
    -\nabla_{\mathbf{\Psi}_\rows}\log \mathrm{pdf}_\mathcal{N}(\mathbf{z} | \mathbf{\Omega}) = \mathbf{X}\mathbf{X}^T - \mathrm{tr}\left[\mathbf{\Omega}^{-1}\left(\mathbf{J}^{ij} \otimes \mathbf{I}\right)\right]_{ij}
    \]
    \[
    -\nabla_{\mathbf{\Psi}_\cols}\log \mathrm{pdf}_\mathcal{N}(\mathbf{z} | \mathbf{\Omega}) =\mathbf{X}^T\mathbf{X} - \mathrm{tr}\left[\mathbf{\Omega}^{-1}\left(\mathbf{I} \otimes \mathbf{J}^{ij}\right)\right]_{ij}
    \]

    It only relies on the data through the terms $\mathbf{X}\mathbf{X}^T$ and $\mathbf{X}^T\mathbf{X}$; we will denote minimum point as $\mathrm{GmGM}\left(\mathbf{X}\mathbf{X}^T, \mathbf{X}^T\mathbf{X}\right)$.
\end{definition*}

\begin{proposition*}[Maximizing $Q$]
    Optimizing $Q$ corresponds to \textbf{the same optimization problem as GmGM}.
    
    GmGM uses the sufficient statistics $\mathbf{X}\mathbf{X}^T$ and $\mathbf{X}\mathbf{X}^T$.  To maximize $Q$, we replace these with:
    
    \[\mathbf{S}^\rows = \mathbb{E}_{\begin{matrix}\scriptscriptstyle\mathbf{z}\sim\mathcal{N}\left(\mathbf{0},\mathbf{\Omega}_{(t)}^{-1}\right)\\\scriptscriptstyle\mathbf{z}\in f^{-1}(\mathbf{y})\end{matrix}}\left[\mathbf{Z}\mathbf{Z}^T\right] \text{ and } \mathbf{S}^\cols = \mathbb{E}_{\begin{matrix}\scriptscriptstyle\mathbf{z}\sim\mathcal{N}\left(\mathbf{0},\mathbf{\Omega}_{(t)}^{-1}\right)\\\scriptscriptstyle\mathbf{z}\in f^{-1}(\mathbf{y})\end{matrix}}\left[\mathbf{Z}^T\mathbf{Z}\right]\]

    More generally, any method to fit unregularized Kronecker-sum-structured models \textbf{without noise} that operates on the sufficient statistics can be used in place of GmGM.
\end{proposition*}
\begin{proof}
    We will first show that the gradient w.r.t. $\mathbf{\Psi}_\rows$ yields a very similar gradient to GmGM; the argument is analogous for $\mathbf{\Psi}_\cols$.  Note also that $Q$ is being maximized, whereas GmGM frames the problem as a minimization problem, so we expect the gradients to have opposite signs.

    \begin{align*}
        \nabla_{\mathbf{\Psi}_\rows} Q\left(\mathbf{\Omega}\Bigl|\mathbf{\Omega}_{(t)}\right) &= \nabla_{\mathbf{\Psi}_\rows}\mathbb{E}_{\mathbf{z}\sim\mathcal{N}\left(\mathbf{0},\mathbf{\Omega}_{(t)}^{-1}\right)\text{ and } \mathbf{z}\in f^{-1}(\mathbf{y})}\left[\log \mathrm{pdf}_\mathcal{N}(\mathbf{z} | \mathbf{\Omega})\right] \\
        &= \mathbb{E}_{\mathbf{z}\sim\mathcal{N}\left(\mathbf{0},\mathbf{\Omega}_{(t)}^{-1}\right)\text{ and } \mathbf{z}\in f^{-1}(\mathbf{y})}\left[\nabla_{\mathbf{\Psi}_\rows}\log \mathrm{pdf}_\mathcal{N}(\mathbf{z} | \mathbf{\Omega})\right] \tag{By dominated convergence theorem} \\
        &= \mathbb{E}_{\mathbf{z}\sim\mathcal{N}\left(\mathbf{0},\mathbf{\Omega}_{(t)}^{-1}\right)\text{ and } \mathbf{z}\in f^{-1}(\mathbf{y})}\left[ \mathrm{tr}\left[\mathbf{\Omega}^{-1}\left(\mathbf{J}^{ij} \otimes \mathbf{I}\right)\right]_{ij} - \mathbf{Z}\mathbf{Z}^T\right] \tag{GmGM definition} \\
        &= \mathrm{tr}\left[\mathbf{\Omega}^{-1}\left(\mathbf{J}^{ij} \otimes \mathbf{I}\right)\right]_{ij} - \mathbb{E}_{\mathbf{z}\sim\mathcal{N}\left(\mathbf{0},\mathbf{\Omega}_{(t)}^{-1}\right)\text{ and } \mathbf{z}\in f^{-1}(\mathbf{y})}\left[ \mathbf{Z}\mathbf{Z}^T\right] \\ 
        &= \mathrm{tr}\left[\mathbf{\Omega}^{-1}\left(\mathbf{J}^{ij} \otimes \mathbf{I}\right)\right]_{ij} - \mathbf{S}_\rows \tag{$\mathbf{S}_\rows$ definition}
    \end{align*}

    It should now be clear by inspection that the maximum of $Q$ corresponds with $\mathrm{GmGM}\left(\mathbf{S}_\rows, \mathbf{S}_\cols\right)$.  This completes the proof.
    
\end{proof}

\section{Efficient evaluation of $\frac{1}{2}\mathrm{tr}\left[\left(\mathbf{P}\mathbf{\Omega}_{(t)}\mathbf{P}^T\right)^{-1}\mathbf{P}\left(\nabla^2 g(\mathbf{z}^*)\right)\mathbf{P}^T\right]$}
\label{apx:med-magma-trace-term}

We will assume $g(\mathbf{z}) = \mathbf{Z}\mathbf{Z}^T$, as $g(\mathbf{z}) = \mathbf{Z}^T\mathbf{Z}$ follows analogously.  As in the previous section, let $\mathbf{J}_{(ij)}$ be the matrix of zeros except for the $(i, j)$ entry being one.  Let us start by deriving the Hessian of $g$, by investigating the directional derivative:

\begin{align*}
    \left(\nabla_{\mathbf{J}_{(k\ell)}}\nabla_{\mathbf{J}_{(ij)}}\right) \hspace{2pt} \mathbf{Z}\mathbf{Z}^T &= \nabla_{\mathbf{J}_{(k\ell)}} \left(\mathbf{Z}\mathbf{J}_{(ij)}^T+\mathbf{J}_{(ij)}\mathbf{Z}^T\right) \\
    &= \mathbf{J}_{(k\ell)}\mathbf{J}_{(ij)}^T + \mathbf{J}_{(ij)}\mathbf{J}_{(k\ell)}^T \\
    &= \mathbf{J}_{(k\ell)}\mathbf{J}_{(ji)} + \mathbf{J}_{(ij)}\mathbf{J}_{(\ell k)} \\
    &= \delta_{j\ell}\left(\mathbf{J}_{(ki)} + \mathbf{J}_{(ik)}\right) \\
    \left[\left(\nabla_{\mathbf{J}_{(k\ell)}}\nabla_{\mathbf{J}_{(ij)}}\right) \hspace{2pt} \mathbf{Z}\mathbf{Z}^T\right]_{ab} &= \delta_{j\ell}\left(\delta_{ak}\delta_{bi} + \delta_{ai}\delta_{bk}\right) \\
    \nabla^2 \mathrm{vec}\left[\mathbf{Z\mathbf{Z}^T}\right]&= \mathbf{I}_{d_\cols} \otimes \left(\mathbf{J}^{(ab)}_{d_\rows} + \mathbf{J}^{(ba)}_{d_\rows}\right) \\
\end{align*}

We will now work out the matrix $\mathbf{P}$ which projects onto the tangent space of the fiber.  As we have seen in Section \ref{apx:med-magma-f-corrupted}, the fiber is linear-in-logspace with $d_\rows + d_\cols - 1$ basis vectors of the form $\mathbf{v}_\cols \oplus \mathbf{v}_\rows$ in which $\mathbf{1}^T\mathbf{v}_\rows= 0 = \mathbf{1}^T\mathbf{v}_\cols$.  Note then that:
\begin{align*}
    \begin{bmatrix}\frac{1}{d_\rows}\mathbf{I}\otimes \mathbf{1}^T \\ \frac{1}{d_\cols}\mathbf{1}^T\otimes \mathbf{I}\end{bmatrix}\left(\mathbf{v}_\cols \oplus \mathbf{v}_\rows\right) &= \begin{bmatrix}\mathbf{v}_\cols + \frac{\mathbf{1}^T\mathbf{v}_\rows}{d_\rows}\mathbf{1} \\ d_\rows \mathbf{v}_\rows + \frac{\mathbf{1}^T\mathbf{v}_\cols}{d_\cols}\mathbf{1}\end{bmatrix} \\
    &= \begin{bmatrix}\mathbf{v}_\cols \\ \mathbf{v}_\rows\end{bmatrix}
\end{align*}

Thus $\mathbf{P} = \begin{bmatrix}\frac{1}{d_\rows}\mathbf{I}\otimes \mathbf{1}^T \\ \frac{1}{d_\cols}\mathbf{1}^T\otimes \mathbf{I}\end{bmatrix}$ projects onto a $d_\rows + d_\cols$ space, from which the tangent space can be reached by just dropping the last row from the matrix.  It should be clear that $\mathbf{P}\nabla^2g(\mathbf{z})\mathbf{P}^T$ is easy, as all matrices involved are conformable block matrices of Kronecker products.  Note that the $\frac{1}{d_\cols}\text{ and }\frac{1}{d_\rows}$ factors cancel out in the trace expression, so we will drop them: $\mathbf{P} = \begin{bmatrix}\mathbf{I}\otimes \mathbf{1}^T \\ \mathbf{1}^T\otimes \mathbf{I}\end{bmatrix}$.

All that remains is to efficiently find the inverse of $\mathbf{P}\mathbf{\Omega}_{(t)}\mathbf{P}^T$.  Note the following:

\begin{align*}
    \mathbf{P}\mathbf{\Omega}_{(t)}\mathbf{P}^T &= \begin{bmatrix}
        \mathbf{\Psi}_\cols \otimes \mathbf{1}^T + \mathbf{I} \otimes \mathbf{1}^T\mathbf{\Psi}_\rows \\ \mathbf{1}^T\mathbf{\Psi}_\cols \otimes \mathbf{I} + \mathbf{1}^T \otimes \mathbf{\Psi}_\rows
    \end{bmatrix}\mathbf{P}^T \\
    &= \begin{bmatrix}
        d_\rows\mathbf{\Psi}_\cols + \left(\mathbf{1}^T\mathbf{\Psi}_\rows\mathbf{1}\right)\mathbf{I} & 
        \mathbf{\Psi}_\cols\mathbf{1} \otimes \mathbf{1}^T + \mathbf{1} \otimes \mathbf{1}^T\mathbf{\Psi}_\rows\\
        \mathbf{1}^T\mathbf{\Psi}_\cols \otimes \mathbf{1} + \mathbf{1}^T \otimes \mathbf{\Psi}_\rows\mathbf{1}& d_\cols\mathbf{\Psi}_\rows + \left(\mathbf{1}^T\mathbf{\Psi}_\cols\mathbf{1}\right)\mathbf{I}
    \end{bmatrix}
\end{align*}

Removing the last row of $\mathbf{P}$ amounts to removing the last row and column of this matrix.  It can easily be inverted by the block matrix inversion formula, requiring inversions of only $d_\ell \times d_\ell$ matrices rather than a $d_\rows d_\cols \times d_\rows d_\cols$ matrix.  Thus, our trace term can be evaluated efficiently.  From this point on, all steps require standard (if tedious) linear algebra, which we leave to the reader.  Our code contains a working implementation for those who wish to reference it.

\section{Description of the systematic approach to preprocessing and evaluating datasets}
\label{apx:med-magma-preprocessing}

As mentioned in the main paper, we systematically downloaded every dataset in the Single Cell Expression Atlas, subject to three criteria (below).  One of the main difficulties in evaluating methods is the ease in which authors can (intentionally or unintentionally) end up cherry-picking examples that make their model look good.  By being systematic, we can guarantee that our results represent the typical outputs of our model.

\paragraph{Criterion 1.} The dataset must be at most 1000 cells large.  We chose this threshold as it was large enough to include many datasets (40 before applying the remaining criteria) while still including only datasets small enough for our algorithm to run quickly on.  It is also a nice round number.  The next `nice round number', in our opinion, is 5000 - which would be 63 datasets.  As our algorithm scales cubically, the largest datasets would take 125 times as long as before.

\paragraph{Criterion 2.} The dataset must contain at least one categorical variable.  This is necessary to have a ground truth to compare against for AMI and assortativity.  When there were multiple, we picked the first one we noticed; these are given in Appendix \ref{apx:med-magma-datasets}.

\paragraph{Criterion 3.} The dataset must clearly reference a source paper in the atlas.  This rule was for our convenience; we did not want to waste time hunting down a paper to cite for the dataset.

30 datasets satisfied all criteria.  We then had to design a standard preprocessing regime that worked for all datasets, choosing to do the minimal preprocessing necessary to get the algorithm to run.  As such, our only requirement was to subset the total number of genes in each dataset; by default, they all had tens of thousands.  We chose to pick the genes so that each dataset became roughly square, to avoid potentially introducing `extreme lopsidedness' as a confounding factor.  Our process was as follows:

\paragraph{Step 1.} Select the top 2000 most variable genes.  This is a very common default preprocessing step and threshold.

\paragraph{Step 2.} Iterate from $s=1\%$ to $s=100\%$ in increments of $1\%$.  For each value of $s$, select only genes with at least $s$ nonzero values.  Whichever value of $s$ leads to the most square matrix, pick those genes.  We chose to pick genes based on sparsity to mirror another standard preprocessing step (throw out `low-quality' genes that do not appear much), and we quantized our values of $s$ to ensure that not all of our matrices were \textit{perfectly} square (which might also confound the results).

\section{Description of datasets used}
\label{apx:med-magma-datasets}

A brief description and citation has been given in Table \ref{tab:med-magma-datasets} for every used dataset.  Note that E-CURD-7 and E-ENAD-21 appear to be the same dataset; it is not a mistake.  Both appear separately in the scRNA-seq database; a different categorical variable has been used for each.

\begin{table}
    \centering
    \begin{tabular}{c|cccc}
        Dataset & Categorical Variable & Categories & Cells & Genes \\\hline
        E-MTAB-4850 \parencite{eltahla_linking_2016} & Sample Characteristic[cell line] & 2 & 60 & 63 \\
        E-GEOD-75367 \parencite{jordan_her2_2016} & Sample Characteristic[facs marker] & 4 & 72 & 70 \\
        E-MTAB-6142 \parencite{karlsson_transcriptomic_2017} & Sample Characteristic[cell cycle phase] & 3 & 96 & 91 \\
        E-GEOD-36552 \parencite{yan_single-cell_2013} & Sample Characteristic[cell type] & 6 & 115 & 127 \\
        E-CURD-10 \parencite{kim_application_2016} & {\small Sample Characteristic[growth condition]} & 3 & 118 & 104 \\
        E-GEOD-99795 \parencite{horning_single-cell_2018} & Factor Value[compound] & 2 & 144 & 128 \\
        E-GEOD-110499 \parencite{fan_linking_2018} & Sample Characteristic[organism part] & 2 & 168 & 169 \\
        E-MTAB-7249 \parencite{nelson_living_2020} & {\tiny Sample Characteristic}{\scriptsize [inferred cell type - authors labels]} & 7 & 185 & 185 \\
        E-MTAB-7606 \parencite{koutsakos_human_2019} & Sample Characteristic[phenotype] & 2 & 209 & 203 \\
        E-GEOD-81383 \parencite{gerber_mapping_2017} & Sample Characteristic[cell line] & 3 & 226 & 226 \\
        E-MTAB-6075 \parencite{fischer_signals_2019} & Factor Value[compound] & 3 & 236 & 231 \\
        E-GEOD-124858 \parencite{liu_single_2019} & Factor Value[passage] & 5 & 241 & 239 \\
        E-MTAB-7381 \parencite{fergusson_maturing_2018} & Sample Characteristic[individual] & 3 & 287 & 291 \\
        E-GEOD-111727 \parencite{golumbeanu_single-cell_2018} & Sample Characteristic[disease] & 2 & 320 & 324 \\
        E-GEOD-109979 \parencite{lu_single-cell_2018} & Factor Value[time] & 4 & 329 & 324 \\
        E-GEOD-100618 \parencite{karamitros_single-cell_2018} & Sample Characteristic[cell type] & 3 & 414 & 418 \\
        E-MTAB-8498 \parencite{chen_re-evaluation_2020} & Sample Characteristic[cell type] & 3 & 448 & 452 \\
        E-MTAB-9801 \parencite{jardine_blood_2021} & {\tiny Sample Characteristic}{\scriptsize [inferred cell type - authors labels]} & 12 & 484 & 501 \\
        E-GEOD-75688 \parencite{chung_single-cell_2017} & Sample Characteristic[histology] & 4 & 520 & 517 \\
        E-GEOD-86618 \parencite{xu_single-cell_2016} & Sample Characteristic[disease] & 2 & 537 & 538 \\
        E-GEOD-98556 \parencite{phillips_novel_2018} & Sample Characteristic[cell type] & 2 & 545 & 550 \\
        E-GEOD-83139 \parencite{wang_single-cell_2016} & {\tiny Sample Characteristic}{\scriptsize [inferred cell type - ontology labels]} & 8 & 617 & 612 \\
        E-CURD-7 \parencite{nguyen_profiling_2018} & Sample Characteristic[sampling site] & 2 & 631 & 627 \\
        E-ENAD-21 \parencite{nguyen_profiling_2018}& Sample Characteristic[cell type] & 2 & 631 & 627 \\
        E-GEOD-70580 \parencite{bjorklund_heterogeneity_2016} & Sample Characteristic[cell type] & 4 & 648 & 643 \\
        E-ENAD-20 \parencite{rambow_toward_2018} & Factor Value[time] & 4 & 667 & 673 \\
        E-GEOD-75140 \parencite{camp_human_2015} & Sample Characteristic[cell type] & 3 & 734 & 728 \\
        E-GEOD-89232 \parencite{breton_human_2016} & Sample Characteristic[cell type] & 2 & 957 & 926 \\
        E-CURD-6 \parencite{zhao_single-cell_2017} & Sample Characteristic[disease] & 4 & 986 & 1031 \\
        E-MTAB-10596 \parencite{hemeryck_organoids_2022} & {\small Sample Characteristic[growth condition]} & 3 & 993 & 984
    \end{tabular}
    \caption{All datasets analyzed during Section \ref{sec:med-magma-systemic}.  The `genes' column counts genes \textit{after} preprocessing.}
    \label{tab:med-magma-datasets}
\end{table}


\end{document}